\documentclass{article}
\usepackage[utf8]{inputenc}
\usepackage{amsmath,amsthm,amssymb}
\usepackage{mathtools, tikz-cd}
\usepackage{color}
\usepackage{bbold}
\usepackage{tabularx}

\theoremstyle{definition}
\newtheorem{definition}{Definition}[section]

\theoremstyle{theorem}
\newtheorem{proposition}{Proposition}[section]

\theoremstyle{definition}

\theoremstyle{remark}

\theoremstyle{definition}
\newtheorem{remark}{Remark}[section]

\newtheorem{example}{Example}[section]

\DeclareMathOperator{\pr}{\operatorname{pr}}

\DeclareMathOperator{\Rbb}{\mathbb{R}}

\DeclareMathOperator{\del}{\partial}

\DeclareMathOperator{\D}{\operatorname{d}}

\DeclareMathOperator{\Pf}{\operatorname{Pf}}

\let\sgn\relax\DeclareMathOperator{\sgn}{\operatorname{sgn}}

\newcommand{\iprod}{\mathbin{\lrcorner}}

\usepackage{pb-diagram}
\usepackage[matrix,arrow,curve]{xy}

\title{Generalized Geometry of 2D Incompressible Fluid Flows}

\author{Radek Suchánek \\ Department of Mathematics and Statistics \\ Masaryk University, Brno}

\date{\today}

\begin{document}

\maketitle

\tableofcontents

\begin{abstract}
We describe a family of generalized almost structures associated with a Monge-Ampère equation for a stream function of a 2D incompressible fluid flows. Using an indefinite metric field constructed from a pair of $2$-forms related to the Monge-Ampère equation, we show the existence of generalized metric compatible structures in our family of generalized structures. Integrability of isotropic structures on the level of Dirac structures and differential forms is discussed.
\end{abstract}


\section{Introduction}


In \cite{Crainic2004GeneralizedCS}, Crainic described conditions on tensor fields to define a generic generalized complex structure \cite{Hitchin-gen-calabi-yau, Gualtierri2011}, and showed how Hitchin pairs of differential forms are connected with these structures. B. Banos showed in \cite{BANOS2007841} that a certain class of non-linear second order PDEs called Monge-Ampère equations gives rise to Hitchin pairs and the corresponding generalized complex geometry and generalized Calabi-Yau geometry \cite{Gualtieri2007BranesOP, Kosmann-Schwarzbach2010, Gualtierri2011}. The notion of integrability then translates to a certain local equivalence problem of Monge-Ampère equations and Jacobi system of non-linear first order PDEs with special non-linearity \cite{Rubtsov2019, Kosmann-Schwarzbach2010}. Motivated by these results, we focus on a different approach towards construction of generalized geometries and show that the family of generalized geometries described by Banos can by extended when dealing with incompressible fluids in 2D. In particular, we prove the existence of generalized Kähler structure and generalized chiral structures, as well as generalized hyper-complex and generalized hyper-para-complex triples. Applications of these structures in the context of twisted supersymmetric non-linear sigma models were studied in \cite{hu2019commuting}. 

The idea behind searching for geometric structures connected with Monge-Ampère equations is two fold. Firstly, if one is interested in concrete examples of generalized structures with specific properties, Monge-Ampère theory provides numerous examples of geometric objects from which the generalized structures can be constructed \cite{kushner_lychagin_rubtsov_2006, Kosmann-Schwarzbach2010, Rubtsov2019}. Secondly, investigating geometry of equations may lead to better understanding and possible improvement of certain models. Example of usefulness of this approach is provided in \cite{Rubtsov2019, roulstone2009kahler, RoubRoul2001, VolRoul2015, phdthesis, mcintyre2001there,Roulstone-Sewell}.


\section{Monge-Ampère theory}\label{section: M-A theory}


All our manifolds will be real and smooth. Let $\mathcal{B}$ be (a real, smooth) two dimensional manifold. Denote by $x,y$ (local) coordinates on $\mathcal{B}$. In the following definition, we will describe a subclass of nonlinear second-order PDEs, with the nonlinearities given by minors and determinant of the hessian matrix.

\begin{definition}
\textit{Monge-Ampère equation on $\mathcal{B}$} is the following nonlinear second-order PDE
    \begin{align}\label{eq: M-A equation in coordinates}
        A f_{xx} + 2B f_{xy} + C f_{yy} + D \left( f_{xx} f_{yy} - {f_{xy}}^2 \right) + E  = 0  \ ,
    \end{align}
where $f_{xy} : = \frac{\partial^2 f}{\partial x \partial y}$, and $A,B,C,D,E$ are smooth functions, which depend on $x,y, f_x, f_y$. 
\end{definition}


\begin{example}
We provide a few standard examples of M-A equations. From \eqref{eq: M-A equation in coordinates} we see that the 2D Laplace equation $f_{xx} + f_{yy} = 0$, as well as the 2D wave equation $f_{xx} - f_{yy} = 0$, are constant coefficients M-A equations. These are particularly interesting in view of proposition \ref{Lychagin-Rubtsov theorem}. Next examples are $\det \operatorname{Hess} (f) = 1$ and $\det \operatorname{Hess} (f) = -1$, which will play a key role in this paper, as they arise from the Navier-Stokes equations for the incompressible 2D fluid flows. Moreover, by virtue of proposition \ref{Lychagin-Rubtsov theorem}, these equations are equivalent with the 2D Laplace and wave equations via a symplectic transformation. Let us end this short overview of M-A equations with the von Karman equation $f_x f_{xx} - f_{yy} = 0$. We will see in example \ref{example: Pfaffian for von Karma} that the von Karman equation has, in a certain sense, very different behaviour then the previous examples. 
\end{example}


Let $T^* \mathcal{B} \xrightarrow{\pi} \mathcal{B}$ be the cotangent bundle over $\mathcal{B}$. Throughout this paper, $\Omega$ will always denote the canonical symplectic form on the cotangent bundle. In the following, we will show how M-A equations can be encoded by a certain pair of differential $2$-forms on $T^* \mathcal{B}$, called \textit{M-A structure}. We start with the following observation. In the Darboux coordinates $x,y,p,q$ ($x,y$ are the base coordinates on $\mathcal{B}$), the canonical symplectic form writes as
    \begin{align}\label{eq: symplectic form}
        \Omega = \D x\wedge \D p + \D y\wedge \D q \ .
    \end{align}
Let $\alpha$ be a $2$-form given by 
    \begin{align}\label{def: general alpha for 2D SMAE}
        \begin{split}
        \alpha & = A \D p \wedge \D y + B( \D x \wedge \D p - \D y \wedge \D q) \\ 
        & \qquad + C \D x \wedge \D q + D \D p \wedge \D q + E \D x \wedge \D y \ , 
        \end{split}
    \end{align} 
where the coefficients $A,B,C,D,E$ are smooth functions on $T^*\mathcal{B}$. Then $\alpha \wedge \Omega = 0$ holds for all possible choice of the coefficients.


\begin{definition}[\cite{kushner_lychagin_rubtsov_2006, Lychagin}]\label{def: M-A structure}
A pair $(\Omega, \alpha) \in \Omega^2 \left( T^*\mathcal{B} \right) \times \Omega^2 \left( T^*\mathcal{B} \right)$ is called \textit{a 2D Monge-Ampère structure over $T^*\mathcal{B}$} (or just \textit{M-A structure}), if $\Omega$ is a symplectic form and the \textit{effectivity condition} $\alpha \wedge \Omega = 0$ is satisfied.
\end{definition}


\begin{remark}\label{remark: notation of coefficients}
Note that the coefficients $A,B, \ldots, E$ in \eqref{def: general alpha for 2D SMAE} are functions on $T^*\mathcal{B}$, while the coefficients $A,B, \ldots, E$ given in \eqref{eq: M-A equation in coordinates} denoted functions on $\mathcal{B}$. We often denote by the same symbol a function $H \in C^\infty \left( T^*\mathcal{B} \right)$, which depends on $x,y,p,q$, and its restriction to the graph of $\D f$, which depends on $x,y,f_x, f_y$. That is, we write $H = H(x,y,p,q)$, as well as $H = H(x,y,f_x, f_y) \equiv (\D f)^* H$. Similarly, we write $F = F(x,y) \in C^\infty\mathcal{B}$ as well as $F = F(x,y,0,0) \in C^\infty \left( T^*\mathcal{B} \right)$. The domain of a function under consideration will always be clear from the context or stated explicitly. 
\end{remark} 

\subsection{From M-A structures to M-A equations}

We will now describe the link between M-A structures and M-A equations \cite{Lychagin-Rubtsov-1983, Lychagin-Rubtsov-Chekalov}. This will justify the nomenclature of definition \ref{def: M-A structure} and also explain the seemingly unnatural factor of $2$ in the second summand of \eqref{eq: M-A equation in coordinates}. Let $f \colon \mathcal{B} \to \mathbb{R}$ be a smooth function. The differential of $f$ determines a section, $\D f \colon \mathcal{B} \to T^*\mathcal{B}$, given by $\D f (x) : = \D_x f$, and we can pullback $\alpha$ onto $\mathcal{B}$ to obtain a top form $(\D f )^* \alpha \in \Omega^2(\mathcal{B})$. Then the equation
    \begin{align}\label{def: M-A equation}
        (\D f )^* \alpha = 0
    \end{align}
defines a nonlinear second order PDE for $f$. If $\alpha$ is given by \eqref{def: general alpha for 2D SMAE}, then the equation \eqref{def: M-A equation} corresponds exactly to \eqref{eq: M-A equation in coordinates}. For a detailed exposition on the theory of Monge-Ampère equations and their applications, particularly in dimensions two and three, see \cite{kushner_lychagin_rubtsov_2006}. For further details about applications of Monge-Ampère theory, for example in the theory of complex differential equations see \cite{BANOS20112187}, theoretical meteorology, and incompressible fluid theory, see \cite{Rubtsov_1997, RoubRoul2001, roulstone2009kahler, phdthesis, VolRoul2015}.


\begin{example}\label{example: differential forms for the Laplace and von Karman eqs}
Let $\alpha = - \D x \wedge \D q + \D y \wedge \D p$. Then 
    \begin{align*}
        (\D f )^* \alpha =  - \D x \wedge \D(f_y) + \D y \wedge \D(f_x) = (-f_{yy} -f_{xx} ) \D x \wedge \D y \ . 
    \end{align*}
Hence $(\D f )^* \alpha = 0$ amounts to the 2D Laplace equation $f_{yy} + f_{xx} = 0$ and $(\Omega, \alpha)$ is the corresponding M-A structure. Now let $\alpha = p \D p \wedge dy + \D x \wedge \D q$. Then $(\D f )^* \alpha = 0$ describes the von Karman equation $f_x f_{xx} - f_{yy} = 0$. 
\end{example}


\begin{definition}[\cite{kushner_lychagin_rubtsov_2006}]\label{def: Pfaffian, elliptic/hyperbolic/normalized structure}
\textit{Pfaffian} of a M-A structure, $\Pf (\alpha) $, is defined by 
    \begin{align}\label{def: Pfaffian}
        \alpha \wedge \alpha = \Pf (\alpha) \Omega \wedge \Omega \ .
    \end{align}
M-A structure is called \textit{non-degenerate}, if the Pfaffian is nowhere-vanishing. If $\Pf (\alpha) > 0$, then she structure is called \textit{elliptic}, and if $\Pf (\alpha) < 0$, then it is called \textit{hyperbolic}. We call a non-degenerate structure \textit{normalized}, if $|\Pf (\alpha)| = 1$. 
\end{definition} 


For a general 2D M-A structure, the Pfaffian in canonical coordinates writes 
    \begin{align}\label{Pfaffian in canonical coordinates}
        \Pf (\alpha) = -B^2 + AC - DE \ .
    \end{align}
    
\begin{example}\label{example: Pfaffian for von Karma}
Let $(\Omega, \alpha)$ be the M-A structure described in example \ref{example: differential forms for the Laplace and von Karman eqs}, i.e. structure for 2D Laplace equation. Then $\operatorname{Pf} (\alpha) = 1$. Now consider the M-A structure for the von Karman equation $f_x f_{xx} - f_{yy} = 0$. Then $\operatorname{Pf} (\alpha) = p$. In section \ref{subsection: endomorphisms and symmetric form}, we will see that the 2D Laplace equation give rise to integrable complex structure, while the von Karman equation does not. 
\end{example}


\begin{remark}
In the modelling of stably stratified geophysical flows, the Pfaffian is related with the the Rellich's parameter \cite{viudez_dritschel_2002, dritschel_viudez_2003, Rubtsov2019}.
\end{remark}


A non-degenerate M-A structure can be normalized $(\Omega, \alpha) \mapsto (\Omega, n(\alpha) )$ by $n (\alpha) : = { |\Pf (\alpha)| }^{ -\frac{1}{2} } \alpha$. Indeed, the Pfaffian of $n (\alpha)$ satisfies $|\Pf \left( n(\alpha) \right)| = 1$, since
    \begin{align*}
        n (\alpha) \wedge n (\alpha) =  { |\Pf (\alpha)| }^{-1} \alpha \wedge \alpha = \frac{ \Pf (\alpha) }{ |\Pf (\alpha)| } \Omega \wedge \Omega \ ,
    \end{align*}
which implies
    \begin{align}\label{eq: Pfaffian of a normalized structure}
        \Pf \left( n(\alpha) \right) = \sgn \Pf (\alpha) \ .
    \end{align}
A non-degenerate M-A structure $(\Omega, \alpha)$ and its normalization $(\Omega, n(\alpha) )$ correspond to the same M-A equation, since 
    \begin{align*}
         (\D f)^* n(\alpha) = \left. { |\Pf (\alpha)| }^{ -\frac{1}{2} }\right\vert_{\operatorname{Im} \D f} (\D f)^* \alpha = 0 \ . 
    \end{align*}


\begin{remark}
Notice that we could have rescaled $\alpha$ with an arbitrary non-vanishing function, which would result in a new M-A structure $(\Omega, \Tilde{\alpha})$. Then, by the same argument as for the normalization, $(\Omega, \Tilde{\alpha})$ defines the same M-A equation as $(\Omega, \alpha)$. The main reason we choose to work with normalized structures is that some of the generalized geometries constructed from $\alpha$ are not invariant with respect to the rescaling. Hence the normalization condition provides a consistent choice of the representative in the class $[\alpha]$, where $\tilde{\alpha} \in [\alpha]$, if and only if $\Tilde{\alpha} = e^{h}\alpha$ for some function $h$.
\end{remark}

\subsection{Endomorphisms and symmetric form}\label{subsection: endomorphisms and symmetric form}

In the following paragraphs, we want to describe how every non-degenerate M-A structure defines a field of endomorphisms, which square to either $\operatorname{Id}_{T \left( T^*\mathcal{B} \right)}$, or $-\operatorname{Id}_{T \left( T^*\mathcal{B} \right)}$. Before we do so, we need to fix some notation.

\hfill 

\noindent \textbf{Notation.} Tensor $\sigma \in \Gamma \left(T^*\mathcal{B} \otimes T^*\mathcal{B} \right)$ can be identified with a $C^\infty(\mathcal{B})$-linear map $\sigma_\# \colon \Gamma \left( T\mathcal{B} \right) \to \Gamma \left( T^*\mathcal{B} \right)$ defined by $ \sigma_\# (X) : = X \iprod \sigma$. Similarly, if $\tau \in \Gamma \left( T\mathcal{B} \otimes T\mathcal{B} \right)$, then we denote by $\tau^\#$ the linear map $\tau^\# \colon \Gamma \left(T^*\mathcal{B} \right) \to \Gamma \left( T\mathcal{B} \right)$, where $\tau^\# (\xi) : = \xi \iprod \tau$. Now consider a non-degenerate $2$-form $\alpha \in \Omega^2(T^*\mathcal{B})$, which means that $\alpha_\# \colon T\mathcal{B} \to T^*\mathcal{B}$ is an isomorphism. Using the inverse $(\alpha_\#)^{-1} \colon T^*\mathcal{B} \to T\mathcal{B}$, we define the bivector $\pi_\alpha \in \Gamma \left( \Lambda^2 T \left( T^*\mathcal{B} \right) \right)$ by 
    \begin{align}\label{def: bivector corresponding to 2-form}
        (\pi_\alpha)^\# : = (\alpha_\#)^{-1} \ . 
    \end{align}
For example, if $\Omega = \D x\wedge \D p + \D y\wedge \D q$ is the canonical symplectic form, then $\pi_\Omega = \del_x \wedge \del_p + \del_y \wedge \del_q $ is the corresponding bivector, where $\partial_x : = \frac{\partial}{\partial_x}$ is the $x$-coordinate vector field (and similarly for the other fields $\partial_i$).

When working with matrices, we will use the underline notation to make the distinction between a morphism and its matrix representation. For example, if $\rho \in \operatorname{End} \left( T \left( T^*\mathcal{B} \right) \right)$, then corresponding matrix will be denoted $\underline{\rho}$ (which will always be understood with respect to the canonical coordinates of $\Omega$). Also, we will be omitting the $\#$ symbol when dealing with the morphisms derived from $2$-forms (and $2$-vectors), e.g. the matrix of $\alpha_\#$ will be denoted simmply $\underline{\alpha}$. Similarly, the matrix of $(\pi_\alpha)^\# $ will be denoted simply by $\underline{\alpha}^{-1}$ (see defining equation \eqref{def: bivector corresponding to 2-form}). When dealing with generalized structures, we write $\mathbb{J}$ for coordinate-free description, as well as for the coordinate description via matrices. Nevertheless there is no space for confusion, since block description of generalized structure either contains coordinate-free objects or their matrices, and the distinction between the two is clear in our notation.

\hfill 

\noindent \textbf{Almost complex and almost product structure.} Let $\rho \in \operatorname{End} \left( T \left( T^*\mathcal{B} \right) \right)$ be an endomorphism defined by
    \begin{align}\label{def: rho_alpha tensor}
        \rho : ={ | \Pf (\alpha) | }^{ - \frac{1}{2} } \pi_\Omega^\# \circ \alpha_\# \ .
    \end{align}
If $(\Omega, \alpha)$ is elliptic, then $\rho^2 = -\operatorname{Id}_{ T \left( T^*\mathcal{B} \right) }$, if $(\Omega, \alpha)$ is hyperbolic, then $\rho^2 = \operatorname{Id}_{ T \left( T^*\mathcal{B} \right) } $. The fact that $\rho$ is either an almost complex structure on $T^*\mathcal{B}$ (if $\Pf (\alpha) > 0$), or an almost product structure (if $\Pf (\alpha) < 0$), was proven in \cite{Lychagin-Rubtsov-1983}. In the canonical coordinates,
    \begin{align}\label{rho structure matrix - general case}
        \underline{\rho} = { | \Pf (\alpha) | }^{- \frac{1}{2} }
            \begin{pmatrix}
            B & -A & 0 & -D \\
            C & -B & D & 0 \\
            0 & E & B & C \\
            -E & 0 & -A & -B 
            \end{pmatrix} 
    \end{align}
and, as expected, it follows that 
    \begin{align*}
        \underline{\rho}^2 = \frac{-\Pf (\alpha)}{ | \Pf (\alpha) | } \operatorname{Id}_{ T \left( T^*\mathcal{B} \right) } = -\sgn \Pf (\alpha) \operatorname{Id}_{ T \left( T^*\mathcal{B} \right) } \ . 
    \end{align*}
Notice that $\rho$ is invariant of the normalization, since \eqref{eq: Pfaffian of a normalized structure} yields
    \begin{align*}
        \sgn \Pf \left( n(\alpha) \right) = \sgn^2 \Pf (\alpha) = \sgn \Pf (\alpha) \ . 
    \end{align*}

\hfill 

\noindent \textbf{Integrability.} V. Lychagin and V. Rubtsov showed in \cite{Lychagin-Rubtsov-1983} that there is a~direct link between local equivalence of M-A equations and integrability of the $\rho$ structure derived from the corresponding M-A structure $(\Omega, \alpha)$. Moreover, the integrability condition can be expressed as a certain closedness condition. 

\begin{proposition}[Lychagin-Rubtsov \cite{Lychagin-Rubtsov-1983}]\label{Lychagin-Rubtsov theorem}
A 2D symplectic M-A equation $(\D f)^* \alpha = 0$ can be locally transformed via a symplectic transformation to either the Laplace equation $\Delta f = 0$, or the wave equation $\Box f = 0$, if and only if the $\rho$ structure is integrable, which is equivalent to $ \frac{ \alpha }{ \sqrt{ | \Pf (\alpha) | } }$ being closed. 
\end{proposition}

\hfill

\noindent \textbf{Symmetric bilinear form.} Every M-A structure defines the following symmetric bilinear form on $T^*\mathcal{B}$ \cite{kushner_lychagin_rubtsov_2006, RoubRoul2001}
    \begin{align}\label{eq: L-R def}
         g (X,Y) : = \frac{  2( X \iprod \alpha \wedge Y \iprod \Omega + Y \iprod \alpha \wedge X \iprod \Omega) \wedge \pi^*\operatorname{vol}  }{\Omega \wedge \Omega} \ ,
    \end{align}
where $X,Y \in \Gamma \left(T (T^* \mathcal{B}) \right)$, $\iprod$ is the interior product, and $\pi^* \operatorname{vol} \in \Omega^2 \left( T^*\mathcal{B} \right)$ is the pullback of a locally chosen top form $\operatorname{vol} \in \Omega^2 \left( \mathcal{B} \right)$ along the cotangent bundle projection $\pi \colon T^* \mathcal{B} \to \mathcal{B} $. Since in our considerations the symplectic form is fixed, we see that the above symmetric field is parameterized by a $2$-form, and the definition \eqref{eq: L-R def} amounts to a mapping $\Omega^2 \left( T^* M \right) \to S^2 \left( T^*\mathcal{B} \right)$, given by $\alpha \mapsto g$. The matrix of $g$ in canonical coordinates is 
    \begin{align}\label{eq: L-R symmetric tensor}
        \underline{g} = 
        \begin{pmatrix}
        2C & -2B & D & 0 \\
        -2B & 2A & 0 & D \\
        D & 0 & 0 & 0 \\
        0 & D & 0 & 0
        \end{pmatrix} \ ,
    \end{align}
where $A,B,C,D$ are the coefficients of $\alpha$ as in \eqref{def: general alpha for 2D SMAE}. Notice the independence of $\underline{g}$ on $E$. Moreover, it is clear that 
    \begin{align*}
        \det \underline{g} \neq 0 \quad \iff \quad D \neq 0 \ .
    \end{align*}
Observe that the coordinate vector fields $\frac{\partial}{\partial p}, \frac{\partial}{\partial q}$ span a $2$-dimensional totally isotropic subspace of $T (T^* \mathcal{B})$. Since $\operatorname{rank} T (T^* \mathcal{B}) = 4$, it follows that $g$ has signature $\left(2,2 \right)$. Applications of the symmetric form \eqref{eq: L-R def} in theoretical meteorology and the corresponding theory of fluid flows were discussed in \cite{VolRoul2015, Lychagin-Rubtsov-1983, RoubRoul2001}. Further information about the properties of $g$, as well as of the corresponding pullback metric $(df)^* g$, in the context of Ricci-flatness condition and Hessian structures was discussed in our previous work \cite{HronekSuchanek}.


\section{Incompressible fluid flows in 2D}


The motivation for us to investigate the stream equation arising from the kinematics of incompressible 2D fluid flows are the works of Banos, Delahaies, Roubtsov, and Roulstone \cite{VolRoul2015, phdthesis, roulstone2009kahler, Rubtsov_1997}. In these papers was demonstrated that modern differential geometric approach offers a new framework for understanding and studying semi-geostrophic theory and related models. In particular, the authors have shown the possibility to model coherent structures in large-scale atmosphero-ocean flows via a hierarchy of approximations to the Navier-Stokes equations and the related incompressible fluid flows models.

\subsection{Kinematics of incompressible fluids} 

Starting point is the system 
    \begin{align}\label{Navier-Stokes}
        \del_t \vec{v} + ( \vec{v} \cdot \vec{\nabla} )  \vec{v} & = \nu \Delta \vec{v} - \vec{\nabla} P \ , \\  \label{incompressibility of the fluid}
        \vec{\nabla} \cdot \vec{v} & = 0
    \end{align}
where $\vec{v}$ is the \textit{velocity field}, $\vec{v} = \left( a_t(x, y), b_t(x, y)  \right)$ ($t$ is understood as a parameter), $\nu \in \Rbb$ is the \textit{viscosity} (which we assume constant), $P = P_t \left(x, y) \right)$ is the \textit{pressure field}, $\cdot$ is the dot product, $\vec{\nabla} = \left(\del_x, \del_y \right) $, and the expression $\Delta \vec{v}$ refers to the vector Laplace operator $\Delta \vec{v} : = \vec{\nabla} \left( \vec{\nabla} \cdot \vec{v} \right) - \vec{\nabla} \times \left( \vec{\nabla} \times \vec{v} \right)$. Since we are working with $2D$ vectors, the vector product is understood as the composition 
    \begin{equation*}
        \begin{tikzcd}
            \Rbb^2 \arrow{r}{\iota} & \Rbb^3 \arrow{r}{\vec{\nabla} \times (-) } & \Rbb^3 \arrow{r}{\vec{\nabla} \times (-) } & \Rbb^3 \arrow{r}{\pr} & \Rbb^2 \ ,  
        \end{tikzcd}
    \end{equation*} 
where $\iota (x, y) =(x, y, 0)$ and $\pr (x, y, z) = (x,y)$. The Navier-Stokes equation \eqref{Navier-Stokes} and the \textit{incompressibility constraint} \eqref{incompressibility of the fluid} together (locally) describe the kinematics of 2D incompressible fluids flows.

\subsection{M-A equation for the stream function} 

Applying $\vec{\nabla} \cdot (-) $ on \eqref{Navier-Stokes} we obtain
    \begin{align}\label{intermediate equation 1}
         \del_t \vec{\nabla} \cdot \vec{v} + \vec{\nabla} \cdot \left( ( \vec{v} \cdot \vec{\nabla} )  \vec{v} \right)  = \nu \vec{\nabla} \cdot \left( \Delta \vec{v} \right) - \Delta P \ .
    \end{align}
Using the incompressibility constraint we have $\vec{\nabla} \cdot \del_t  \vec{v} = \del_t \vec{\nabla} \cdot \vec{v} = 0$ and     
    \begin{align*}
        \vec{\nabla} \cdot \left( ( \vec{v} \cdot \vec{\nabla} )  \vec{v} \right) & =  \left( \del_x a \right)^2 + \left( \del_y b \right)^2 + 2 \del_y a \del_x b \ , \\
            & = - 2 \del_x a \del_y b + 2 \del_y a \del_x b \ . 
    \end{align*}
Hence the left-hand side of the equation \eqref{intermediate equation 1} can be written as
    \begin{align*}
        \vec{\nabla} \cdot \del_t \vec{v} + \vec{\nabla} \cdot \left( ( \vec{v} \cdot \vec{\nabla} )  \vec{v} \right) = - 2 \det 
            \begin{pmatrix}
            \del_x a & \del_y a \\
            \del_x b & \del_y b 
            \end{pmatrix}
    \end{align*}
The term $\nu \vec{\nabla} \cdot \left( \Delta \vec{v} \right)$, appearing on the right-hand side of \eqref{intermediate equation 1}, vanish since
    \begin{align*}
        \vec{\nabla} \cdot \left( \Delta \vec{v} \right) =  \vec{\nabla} \cdot \left[ \vec{\nabla} \times \left( \vec{\nabla} \times \vec{v} \right) \right] = \del_{xxy} b - \del_{xyy}a + \del_{xyy} a - \del_{xxy}b = 0
    \end{align*}
The above computation leads to the following Poisson equation as a consequence of the system of equations \eqref{Navier-Stokes}, \eqref{incompressibility of the fluid}
    \begin{align}\label{intermediate equation 2}
         \Delta P = 2 \det 
            \begin{pmatrix}
            \del_x a & \del_y a \\
            \del_x b & \del_y b 
            \end{pmatrix} \ .
    \end{align}
The fluid incompressibility equation \eqref{incompressibility of the fluid} can be solved by $\vec{v} = \left( f_y, - f_x \right) $. This implies that \eqref{intermediate equation 2} can be written as 
    \begin{align}\label{Monge-Ampère equation for the stream function}
        \det \operatorname{Hess} (f) = \frac{1}{2} \Delta P \ . 
    \end{align}
Instead of viewing the above equation as a Poisson equation for $P$, one can instead interprete it as a Monge-Ampère equation for a function $f$ of two variables, $f = f(x,y)$, called the \textit{stream function}. This equation was studied by Roulston, Banos, Gibbon, and Roubtsov, in the context of semigeostropic models \cite{roulstone2009kahler, coherent_structures_Navier-Stokes}, where the relationship between $\sgn \Delta P$ and the balance between the rate of strain and the enstrophy of the flow was discussed. In order to proceed further with our geometric approach, we assume $\Delta P \neq 0$.

\subsection{M-A geometry of incompressible fluid flows}\label{section: M-A geometry of incompressible fluid flows}

We now move toward the description of the M-A equation \eqref{Monge-Ampère equation for the stream function} via a pair of $2$-forms $(\Omega, \alpha)$ forming a normalized M-A structure. The corresponding pair of differential $2$-forms is given by the canonical symplectic form \eqref{eq: symplectic form} and 
    \begin{align}\label{Monge-Ampere structure}
        \alpha = \sqrt{ \frac{ 2 }{ |\Delta P| } } \D p \wedge \D q  - \frac{ \Delta P }{ \sqrt{ |2 \Delta P| } } \D x \wedge \D y \ ,
    \end{align}
where $\Delta P \neq 0$ is viewed as a function on $T^*\mathcal{B}$, which coincides with $\Delta P$ from equation \eqref{Navier-Stokes} (see remark \ref{remark: notation of coefficients}). The pair clearly satisfy $ \alpha \wedge \Omega = 0$, making it an M-A structure. This is in concordance with general situation described in section \ref{section: M-A theory}. Moreover,
    \begin{align*}
        (\D f)^*\alpha 
            & = \left( \sqrt{ \frac{ 2 }{ |\Delta P| } } ( f_{xx} f_{yy} - { f_{xy} }^2 ) - \frac{ \Delta P }{ \sqrt{ |2\Delta P| } }  \right) \D x \wedge \D y  \ .
    \end{align*}
This implies that $ (\D f)^*\alpha = 0$ amounts to $ \det \operatorname{Hess} (f) = \frac{\Delta P}{2}$, and thus the above M-A structure corresponds to the M-A equation \eqref{Monge-Ampère equation for the stream function}. From \eqref{Monge-Ampere structure} we read the coefficients of $\alpha$
    \begin{align*}
        A = 0 \ , B = 0 \ , C = 0 \ , D = \sqrt{ \frac{ 2 }{ |\Delta P| } } \ , E = \frac{ - \Delta P }{ \sqrt{ |2 \Delta P| } } \ .
    \end{align*}
By direct computation, or with the help of \eqref{Pfaffian in canonical coordinates}, we see that the Pfaffian is 
    \begin{align*}
        \Pf (\alpha) = \sgn \Delta P \ .
    \end{align*}
Hence the M-A structure is normalized in the sense of definition \ref{def: Pfaffian, elliptic/hyperbolic/normalized structure}. Consequently, the matrix of endomorphism $\rho$ is
    \begin{align}\label{matrix of rho for det hess = 1}
        \underline{\rho} = { |2 \Delta P| }^{ - \frac{1}{2} }
            \begin{pmatrix}
            0 & 0 & 0 & - 2 \\
            0 & 0 & 2 & 0 \\
            0 &  - \Delta P & 0 & 0 \\
             \Delta P & 0 & 0 & 0 
            \end{pmatrix} \ . 
    \end{align}
Depending on the $\sgn \Delta P $, we obtain the following cases
    \begin{align}\label{eq: square of rho tensor for det hess = 1}
        \rho^2 = 
        \begin{cases}
        - \operatorname{Id}_{ T \left( T^*\mathcal{B} \right) } & \ \ \ \text{for $\Delta P > 0$} \ ,  \\
        \ \ \  \operatorname{Id}_{ T \left( T^*\mathcal{B} \right) } & \ \ \ \text{for $\Delta P < 0$} \ , 
        \end{cases}
    \end{align}


Let us proceed with the question of integrability. By virtue of proposition \eqref{Lychagin-Rubtsov theorem}, the integrability of $\rho$ is equivalent to $\D  \left( { | \Pf (\alpha) | }^{ - \frac{1}{2} } \alpha \right) = 0$, where $\alpha$ is given by \eqref{Monge-Ampere structure}. Recall that the pressure is assumed to be a function of $x,y$ only, $P = P(x,y)$, hence we have
    \begin{align*}
        \D  \left( { | \Pf (\alpha) | }^{ -\frac{1}{2} } \alpha \right) = \left( -\frac{ \sqrt{2} }{ 2 } |\Delta P|^{-\frac{ 3 }{ 2 } } \sgn \Delta P \right) \D \left( \Delta P \right) \wedge \D p \wedge \D q \ . 
    \end{align*}
Since $\Delta P$ is a $0$-form (hence cannot be closed), the last equation implies
    \begin{align*}
        \rho \text{ is integrable} \quad \iff \quad \Delta P \text{ is constant.}
    \end{align*}

\hfill

\noindent \textbf{Indefinite metric field.} To finish this section, we give the matrix description of the symmetric bilinear form \eqref{eq: L-R def} corresponding to \eqref{Monge-Ampère equation for the stream function}. The matrix of $g$ for a generic $\alpha$ is displayed in \eqref{eq: L-R symmetric tensor}, which in our case becomes the following pseudo-metric field with split signature $(2,2)$
    \begin{align}\label{eq: Lychagin-Rubtsov metric for det hess = 1}
        \underline{g} = \sqrt{ \frac{2}{ |\Delta P| } }
            \begin{pmatrix}
            0 & \mathbb{1} \\
            \mathbb{1} & 0
            \end{pmatrix} \ .
    \end{align}



\section{Generalized geometry and M-A theory}


In this section we want to describe generalized structures associated with the equation for the stream function \eqref{Monge-Ampère equation for the stream function}. For the sake of completeness, we recall some key definitions and objects from generalized geometry. To avoid any confusion, we formulate these notions for a general smooth manifold $\mathcal{M}$, and then choose  $\mathcal{M} = T^*\mathcal{B}$ to investigate the generalized geometry of $\det \operatorname{Hess} ( f ) = 1$.

\subsection{Generalized (almost) structures} 

A \textit{generalized tangent bundle} over $\mathcal{M}$ is the vector bundle 
    \begin{equation}\label{generalized tangent bundle}
        \begin{tikzcd}
             \mathbb{T}\mathcal{M} : = T\mathcal{M} \oplus T^*\mathcal{M} \arrow{r}{\pi} &  \mathcal{M} \ ,  
        \end{tikzcd}
    \end{equation} 
with the bundle projection $\pi$ defined by the the composition 
    \begin{equation*}
        \begin{tikzcd}
             \mathbb{T}\mathcal{M} \arrow{r}{} &  T\mathcal{M} \arrow{r}{} & \mathcal{M} \ ,  
        \end{tikzcd}
    \end{equation*} 
where the left map is the projection on the first factor of the Whitney sum, and the second map is the tangent bundle projection on $\mathcal{M}$. The pairing between vector fields and $1$-forms endows $\mathbb{T}\mathcal{M}$ with a non-degenerate, symmetric, bilinear form $\eta$
    \begin{align}\label{def: inner product on the generalized tangent bundle}
        \eta \bigl( \left( X, \xi \right),  \left( Y, \zeta \right) \bigr) : = \frac{1}{2} \left( \xi (Y) + \zeta (X) \right) \ ,
    \end{align}
which defines on $\mathbb{T}\mathcal{M}$ a pseudo-Riemannian metric of signature $(n,n)$. If we choose a coordinate system $(q^\mu)$ on $\mathcal{M}$, then the corresponding coordinate vector fields and $1$-forms define a local basis $\left( (\partial_{q^\mu}, 0 ) , ( 0, dq^\mu ) \right) $ of $\mathbb{T} \mathcal{M}$. The matrix representation of $\eta$ in this basis is 
    \begin{align*}
        \eta = 
            \begin{pmatrix}
            0 & \mathbb{1} \\
            \mathbb{1} & 0
            \end{pmatrix} \ .
    \end{align*}
    

\textit{A generalized almost complex structure $\mathbb{J}$ over $\mathcal{M}$} is a bundle map $\mathbb{J} \colon \mathbb{T}\mathcal{M} \to \mathbb{T}\mathcal{M}$ such that $\mathbb{J}^2 = - \operatorname{Id}_{ \mathbb{T}\mathcal{M}}$ and for all $ \left( X, \eta \right),  \left( Y, \theta \right) \in \mathbb{T}\mathcal{M}$
    \begin{align}\label{eq: J acting on eta}
         \langle \mathbb{J} \left( X, \eta \right),  \mathbb{J} \left( Y, \theta \right) \rangle = \langle  \left( X, \eta \right), \left( Y, \theta \right) \rangle
    \end{align}
where $\eta$ is the natural inner product \eqref{def: inner product on the generalized tangent bundle} \cite{Gualtierri2011, Hitchin-gen-calabi-yau}.


\begin{example}\label{example: Hitchin pairs and generalized complex structure}
Let $(\Omega, \alpha)$ be a non-degenerate M-A structure with $\alpha$ closed (i.e. a pair of symplectic structures). Consider a $(1,1)$-tensor $A_\alpha : =  \pi_\Omega^\# \circ \alpha_\#$ (c.f. definition \eqref{def: rho_alpha tensor}). Then $(\Omega, A_\alpha)$ defines \textit{a Hitchin pair of $2$-forms} in the sense of Crainic \cite{Crainic2004GeneralizedCS, Kosmann-Schwarzbach2010} since 
    \begin{align*}
        \Omega_\# A_\alpha = A_\alpha^* \Omega_\# \ .  
    \end{align*}
Using the notion of Hitchin pairs, B. Banos \cite{BANOS2007841} showed in that every 2D non-degenerate M-A structure satisfying for appropriate $\phi \in C^\infty \left( T^* \mathcal{B} \right) $ the \textit{divergence condition}
    \begin{align*}
        \D (\alpha + \phi \Omega ) = 0 \ , 
    \end{align*}
yields an integrable generalized almost structure $\mathbb{J}_\alpha$ given as follows
    \begin{align*}
        \mathbb{J}_\alpha = 
            \begin{pmatrix}
            A_\alpha & \pi_\Omega^\# \\
            - \left( \Omega_\# + \Omega_\# A_\alpha^2 \right) & -A_\alpha^*
            \end{pmatrix}
    \end{align*}
\end{example}


The result described in example \ref{example: Hitchin pairs and generalized complex structure} was a key motivation for our further investigations of other possiblities of constructing generalized (almost) geometries from Monge-Ampère equations and the corresponding Monge-Ampère structures. Starting with the M-A equation for the stream function \eqref{Monge-Ampère equation for the stream function} with specifically chosen $\Delta P$, we will show that the family of generalized structures constructed with the help of M-A theory can be significantly enlarged. In order to state our result, we need to extend the notion of generalized (almost) complex structure by the following definition. 


\begin{definition}[\cite{hu2019commuting}] \label{def: generalized almost structure}
\textit{A generalized almost structure $\mathbb{J}$ over $\mathcal{M}$} is a bundle map $\mathbb{J} \colon \mathbb{T}\mathcal{M} \to \mathbb{T}\mathcal{M}$ such that
    \begin{align*}
        \mathbb{J}^2 =  \gamma_1 \operatorname{id}_{\mathbb{T}M} \ , && \mathbb{J}^\bullet \eta = \gamma_2 \eta \ ,
    \end{align*}
where $\gamma_1, \gamma_2 \in \{ -1, 1\}$ and $\mathbb{J}^\bullet\eta$ is understood in the sense of \eqref{eq: J acting on eta}. Table \ref{tab: type of generalized structure} describes the four possible choice of constants $\gamma_i$.
    \begin{table}[h!]
        \begin{center}
        \begin{tabular}{ | >{\centering\arraybackslash}X p{2cm} || >{\centering\arraybackslash}X p{2cm}| >{\centering\arraybackslash}X p{2cm} |  >{\centering\arraybackslash}X p{2cm} |  >{\centering\arraybackslash}X p{2cm} |} 
            \hline 
            $(\gamma_1, \gamma_2)$ & $(1, 1)$ & $(1, -1)$ & $(-1, 1)$ & $(-1, -1)$ \\ 
            \hline
            type of $\mathbb{J}$  &  GaP & GaPC & GaC & GaAC  \\
            \hline
        \end{tabular}
        \caption{\label{tab: type of generalized structure} Type of a generalized almost structure depending on $(\gamma_1, \gamma_2)$.}
        \end{center}
    \end{table}
The abbreviations stand for \textit{generalized almost product} (GaP), \textit{generalized almost complex} (GaC), \textit{generalized almost para-complex} (GaPC), and \textit{generalized almost anti-complex} (GaAC) structure. A generalized structure is called \textit{non-degenerate}, if its eigenbundles are isomorphic to $T \mathcal{M}$ (if $\mathbb{J}^2 = \operatorname{Id}_{ \mathbb{T} \mathcal{M} }$), or to $T \mathcal{M} \otimes \mathbb{C}$ (if $\mathbb{J}^2 = - \operatorname{Id}_{ \mathbb{T} \mathcal{M} }$)
\end{definition}


\begin{example}\label{diagonal and antidiag. g. almost structures}
A family of examples of generalized almost structures comes from  non-degenerate tensor fields of rank $2$.  Let $J \in \operatorname{Aut} \left( T \mathcal{M} \right) $ be an almost complex structure, $\alpha \in \Omega^2 \left( \mathcal{M} \right)$ a non-degenerate $2$-form, and $g \in S^2 \left( \mathcal{M} \right)$ a~non-degenerate symmetric bilinear form. We have seen in the previous sections that M-A theory is a natural source of the above three object. Let $\epsilon \in \{ -1, 1 \}$ and consider
    \begin{align}\label{diagonal and antidiagonal structures}
        \mathbb{J}  & = 
            \begin{pmatrix}
            J & 0 \\
            0 & \epsilon J^*
            \end{pmatrix} \ , &
        \mathbb{K} & = 
            \begin{pmatrix}
            0 & \pi_\alpha^\# \\
            \epsilon \alpha_\# & 0
            \end{pmatrix} \ , & 
        \mathbb{G} & = 
            \begin{pmatrix}
            0 & \pi_g^\# \\
            \epsilon g_\# & 0
            \end{pmatrix} \ . \
    \end{align}
Then $\mathbb{J}$ is a GaAC stucture for $\epsilon = 1$, and a GaC structure for $\epsilon = -1$, while $\mathbb{K}$ is a GaPC structure for $\epsilon = 1$, and a GaC structure for $\epsilon = -1$. For $\mathbb{G}$ we have a GaP in the case $\epsilon = 1$ and a GaAC if $\epsilon = -1$. 

We see that starting with non-degenerate rank $2$ tensor fields, one can construct isotropic, as well as non-isotropic structures. We will end this example with the demonstration of non-isotropy of generalized almost product structure. Take $\mathbb{G}$ with $\epsilon = 1$. Then the $\pm1$-eigenbundles of $\mathbb{G}$ are
    \begin{align*}
        E_\pm = \{ (X, \pm g_\# X ) | X \in \Gamma \left( T \mathcal{M} \right) \} \ .
    \end{align*}
Indeed, for $x_+ = (X, g_\# X ) \in E_+$, and $y_- = (Y, -g_\# Y ) \in E_-$, it holds $\mathbb{G} x_+ = (X, g_\# X ) = x_+$, and $\mathbb{G} y = (-Y, g_\# Y ) = -y_-$. Now suppose $x_+,y_+ \in E_+$. Then 
    \begin{align*}
        \eta \left( x_+ , y_+ \right) = \frac{1}{2}\left( (g_\# X) Y + (g_\# Y) X \right) = g(X,Y) \ .
    \end{align*}
Obviously there are $X,Y \in \Gamma \left( T \mathcal{M} \right)$ such that $g(X,Y) \neq 0$, thus $E_+$ is not totally isotropic, and hence $\mathbb{G}$ is a non-isotropic structure. Nontheless, we can still speak about (non-)degeneracy of a non-isotropic structures. 
\end{example}

\hfill

\noindent \textbf{The Courant bracket.} The space of sections $\Gamma \left(\mathbb{T}\mathcal{M} \right) = \Gamma \left( T\mathcal{M} \right) \oplus \Gamma \left( T^*\mathcal{M} \right)$ is equipped with the antisymmetric \textit{Courant bracket} $[ -,- ]_C$ \cite{Courant-DiracBracket}
    \begin{align}\label{def: Courant bracket}
       [ \left( X, \xi \right),  \left( Y, \zeta \right) ]_C : = \left( [X,Y] , \mathcal{L}_X \zeta -  \mathcal{L}_Y \xi - \frac{1}{2} \D ( X \iprod \zeta -  Y \iprod \xi ) \right) \ ,
    \end{align}
where $[X,Y]$ is the Lie bracket and $\mathcal{L}$ is the Lie derivative. Note that the Courant bracket, which is an extension of the Lie bracket for vector fields, does not satisfy the Jacobi identity. More importantly for our considerations, the Courant bracket can be used to define integrability of \textit{isotropic} structures on $\mathbb{T}\mathcal{M}$.

\subsection{Isotropy, Dirac structures, and integrability}

Every generalized almost structure comes together with the corresponding subbundles $E_{+} , E_{-}$, which are $\pm 1$-eigenbundles if $\mathbb{J}^2 = \operatorname{id}_{\mathbb{T}M}$, and $\pm i$-eigenbundles if $\mathbb{J}^2 = -\operatorname{Id}_{\mathbb{T}M}$. 

\begin{definition}
A subbundle $E \subset \mathbb{T}M$ is called \textit{totally isotropic} (w.r.t. the inner product $\eta$), if for all $x,y \in E: \eta(x,y) = 0$. Totally isotropic $E$ is called an \textit{almost Dirac structure on $M$} if $\operatorname{rank}E = \operatorname{rank}TM$. A \textit{Dirac structure on $M$} is an almost Dirac structure $E$ such that $[E, E]_C \subset E$.
\end{definition}

\begin{remark}
A Dirac structure on $M$ can be equivalently defined as a totally isotropic subbundle $E \subset \mathbb{T}M$ of maximal rank, which is involutive with respect to the Courant bracket.      
\end{remark}

The four possible generalized almost structures determined by definition \ref{def: generalized almost structure} can be divided into two subsets depending on whether the eigenbundles $E_\pm$ are almost Dirac structures or not. 

\begin{definition}
Let $\mathbb{J}$ be a generalized almost structure. If the eigenbundles $E_\pm$ are almost Dirac structures, then $\mathbb{J}$ is called \textit{isotropic} (with respect to $\eta$). Otherwise $\mathbb{J}$ is called \textit{non-isotropic}.
\end{definition}

The involutivity condition of almost Dirac structures $E_\pm$ can serve as a~definition of integrability only for isotropic $\mathbb{J}$ (see remark \ref{non-isotropy and integrability} below).

\begin{definition}\label{def: integrability of a generalized structure}
An isotropic generalized almost structure $\mathbb{J} \in \operatorname{End} (\mathbb{T}M) $ is called \textit{integrable} if the corresponding eigenbundles $E_+, E_-$ are Dirac structures. An integrable generalized almost structure is called a \textit{generalized structure}. 
\end{definition}

Let $\mathbb{J} \in \operatorname{End} ( \mathbb{T}M) $ be an isotropic generalized almost structure. Then the torsion of $\mathbb{J}$ is defined for all $x,y \in \Gamma ( \mathbb{T}M )$ by 
    \begin{align}\label{def: torsion of A}
        N_{\mathbb{J}} (x,y) : = [\mathbb{J}x, \mathbb{J}y]_C + \mathbb{J}^2 [x,y]_C - \mathbb{J} \left( [ \mathbb{J}x, y ]_C + [x, \mathbb{J}y]_C \right) \ .
    \end{align}
This $(1,2)$-tensor $N_{\mathbb{J}} \colon \Gamma ( \mathbb{T}M ) \otimes \Gamma ( \mathbb{T}M ) \to \Gamma ( \mathbb{T}M ) $ is called (generalized) \textit{Nijenhuis tensor}. An isotropic $\mathbb{A}$ is integrable if and only if the coresponding Nijenhuis tensor vanishes: $N_{\mathbb{A}} (x,y) = 0$ for all $x,y \in \Gamma (\mathbb{T}M)$ \cite{Crainic2004GeneralizedCS}. 

\begin{remark}\label{non-isotropy and integrability}
If a generalized almost structure $\mathbb{J}$ is non-isotropic, then the eigenbundles $E_\pm$ are not totally isotropic (see example \ref{diagonal and antidiag. g. almost structures}) and the Courant bracket is not well-defined on them. Moreover, for non-isotropic structures, the torsion \eqref{def: torsion of A} is not a tensor. As a consequence, the notion of integrability, which is usually given either by the condition of vanishing Nijenhuis tensor, or by definition \ref{def: integrability of a generalized structure}, cannot be applied for non-isotropic structures. Nontheless, there are other means of defining integrability of non-isotropic structures. For example, the authors of \cite{hu2019commuting} defined \textit{weak integrablity} using the notion of generalized Bismut connection \cite{Gualtieri2007BranesOP}.
\end{remark}

\section{Generalized geometry of 2D incompressible fluids}

We have seen in section \ref{section: M-A geometry of incompressible fluid flows} that, by virtue of the proposition \ref{Lychagin-Rubtsov theorem}, the endomorphism $\rho$ is an integrable structure, if and only if $\Delta P$ is constant, nowhere vanishing function. This also implies $\D \alpha = 0$. In the following proposition, we will focus on the case $\Delta P > 0$ and, for the convenience choose $\Delta P = 2$. In this situation, the equation for the stream function \eqref{Monge-Ampère equation for the stream function} becomes
    \begin{align}\label{Monge-Ampère equation for the stream function with Delta P = 2}
        \det \operatorname{Hess} (f) = 1 \ .
    \end{align}
The block matrices of $\rho$ and $\alpha_\#$ are 
    $\underline{\rho} = 
        \begin{pmatrix}
        0 & A \\
        A & 0
        \end{pmatrix}$ and 
    $\underline{\alpha} = 
        \begin{pmatrix}
        A & 0 \\
        0 & -A
        \end{pmatrix}$, where 
    $A = 
        \begin{pmatrix}
        0 & -1 \\ 
        1 & 0\end{pmatrix}$.
We have $ \underline{\rho}^2 = \underline{\alpha}^2 = - \mathbb{1}$, and $\underline{\rho}^T = - \underline{\rho}$. Notice that this relations change when $\Delta P = -2$. Of course, in all cases we have $\underline{\alpha}^T = - \underline{\alpha}$, since $\alpha$ is a $2$-form.


\begin{proposition}
Let $(\Omega, \alpha)$ be the normalized M-A structure determined by the equation \eqref{Monge-Ampère equation for the stream function} with $\Delta P = 2$. Let $\rho$ be the corresponding endomorphism \eqref{def: rho_alpha tensor}. Then the following $\mathbb{J}_i \colon \mathbb{T} \mathcal{B} \to \mathbb{T} \mathcal{B}$
    \begin{align}\label{triple of gen. structures 1}
        \mathbb{J}_1 = 
            \begin{pmatrix}
            \underline{\rho} & 0 \\
            0 & \epsilon_1 \underline{\rho}
            \end{pmatrix} \ ,
        && \mathbb{J}_2 = 
            \begin{pmatrix}
            0 & \underline{\alpha} \\
            \epsilon_2 \underline{\alpha} & 0
            \end{pmatrix} \ , 
        && \mathbb{J}_3 =
            \begin{pmatrix}
            0 & \underline{\Omega} \\
            \epsilon_3 \underline{\Omega} & 0
            \end{pmatrix}  \ ,
    \end{align}
where $\epsilon_i \in \{ -1, 1\}$, are generalized almost structures structures. 
    \begin{enumerate}
        \item If $\epsilon_1 = 1$, then $\mathbb{J}_1$ is isotropic and integrable. If $\epsilon_1 = -1$, then $\mathbb{J}_1$ is non-isotropic.
        \item  The structures $\mathbb{J}_2, \mathbb{J}_3$ are isotropic and integrable.
    \end{enumerate}
The types of these structures depending on $\epsilon_i$ are described in table \ref{tab: type of generalized structure for det hess = 1}. 
    \begin{table}[h!]
        \begin{center}
        \begin{tabular}{ | >{\centering\arraybackslash}X p{2cm} | >{\centering\arraybackslash}X p{2cm}| >{\centering\arraybackslash}X p{2cm} |  >{\centering\arraybackslash}X p{2cm} | } 
            \hline 
            $\sgn{\epsilon_i}$    & $\mathbb{J}_1$ & $\mathbb{J}_2 $ & $\mathbb{J}_3$\\ 
            \hline
            $+$  &  GC & GC & GC \\
            \hline
            $-$  & GaAC & GPC & GPC \\
            \hline
        \end{tabular}
        \caption{\label{tab: type of generalized structure for det hess = 1} Generalized structure associated with $\det \operatorname{Hess} (f) = 1$.}
        \end{center}
    \end{table}
\end{proposition}

\begin{proof}
The choice of the triple $(\mathbb{J}_i)$ corresponds to the example \eqref{diagonal and antidiag. g. almost structures} with respect to the M-A equation $\det \operatorname{Hess} (f) = 1$. This implies 
    \begin{align}\label{eq: (J_i)^2}
        (\mathbb{J}_1)^2 = -\mathbb{1} \ , &&  (\mathbb{J}_2)^2 = \epsilon_2 (-\mathbb{1}) \ , &&  (\mathbb{J}_3)^2 = \epsilon_3 (-\mathbb{1}) \ , 
    \end{align}
and 
    \begin{align}
        (\mathbb{J}_1)^* \eta = \epsilon_1 \eta \ , && (\mathbb{J}_2)^* \eta = \epsilon_2 \eta \ , && (\mathbb{J}_3)^* \eta = \epsilon_3 \eta \ ,
    \end{align}
where $(\mathbb{J}_i)^* \underline{\eta} = {\mathbb{J}_i}^T \underline{\eta} \mathbb{J}_i$. Due to our choice of $\Delta P$, $\rho$ is an integrable almost complex structure (a complex structure). If $\epsilon_1 = 1$, this implies that $\mathbb{J}_1$ is a generalized complex structure. If $\epsilon_1 = -1$, then $\mathbb{J}_1$ is a generalized almost anti-complex structure, which is non-isotropic. Furthermore, both $\alpha$ and $\Omega$ are closed. Thus for all choices of $\epsilon_2, \epsilon_3$, the structures $\mathbb{J}_2$ and $\mathbb{J}_3$ are integrable. Therefore $\mathbb{J}_2$ and $\mathbb{J}_3$ are (independently) either GC or GPC structures. 
\end{proof}

\subsection{Structures generated by anticommutative triples}

Let $\{  \mathbb{J}_1, \mathbb{J}_2 , \mathbb{J}_3 \}$ be a triple of pair-wise anticommuting generalized almost structures, given by \eqref{triple of gen. structures 1}. The anticommutators are 
    \begin{align*}
        \{\mathbb{J}_1, \mathbb{J}_2\} & =  
            \begin{pmatrix}
            0 & -( \epsilon_1 - 1 ) \underline{\Omega} \\
            \epsilon_2 (\epsilon_1 - 1 ) \underline{\Omega} & 0
            \end{pmatrix}   \ , \\
        \{\mathbb{J}_2, \mathbb{J}_3\} & =   
            \begin{pmatrix}
            -(\epsilon_2 - \epsilon_3 ) \underline{\rho} & 0 \\
            0 &  (\epsilon_2 - \epsilon_3 ) \underline{\rho} 
            \end{pmatrix}  \ , \\
        \{\mathbb{J}_1, \mathbb{J}_3\} & =  
            \begin{pmatrix}
            0 & -(1 - \epsilon_1) \underline{\alpha} \\
             \epsilon_3 (1 - \epsilon_1) \underline{\alpha} & 0
            \end{pmatrix}   \ ,
    \end{align*}
and we observe 
    \begin{align*}
        \{\mathbb{J}_i, \mathbb{J}_j\} & = 0 \ \forall i,j \quad \iff \quad \epsilon_1 = 1,  \epsilon_2 = \epsilon_3 \ .
    \end{align*}
This implies that there are exactly two possibilities in which the pair-wise anticommutivity occurs. In both cases $\mathbb{J}_1$ is a GC structure, and either  $\mathbb{J}_2, \mathbb{J}_3$ are a pair of GC structures, or a pair of GPC structures. The first case consisting of a triple of GC structures has some further properties which are described in the proposition \ref{prop: the 2-sphere of almost GC structures} and the subsequent corollary. In both cases, any of the two structures generate the third one and $\mathbb{J}_1$ is a fixed GC structure.

\begin{proposition}\label{prop: the 2-sphere of almost GC structures}
Let $\{ \mathbb{J}_1, \mathbb{J}_2, \mathbb{J}_3 \}$ be given by \eqref{triple of gen. structures 1}. Then the structures pair-wise anticommute, 
    \begin{align*}
         \{ \mathbb{J}_i, \mathbb{J}_{i+1} \} = 0 \ , \  i = 1,2,3 \ ,
    \end{align*}
if, and only if, $\epsilon_1 = \epsilon_2 = \epsilon_3 = 1$, or, $\epsilon_1 = - \epsilon_2 = - \epsilon_3 = 1$. 
    \begin{enumerate}
    \item Let $\epsilon_1 = \epsilon_2 = \epsilon_3 = 1$. Then each $\mathbb{J}_i$ is a GC structure. Moreover, there is a $2$-sphere of GaC structures in $\operatorname{span}_{\mathbb{R}} \{ \mathbb{J}_1, \mathbb{J}_2, \mathbb{J}_3 \}$.
    \item Let $\epsilon_1 = - \epsilon_2 = - \epsilon_3 = 1$. Then $\mathbb{J}_1$ is a GC structure and $\mathbb{J}_2, \mathbb{J}_3$ are GPC structures. Moreover, there is a one-sheeted hyperboloid of GaPC structures and a two-sheeted hyperboloid of GaC structures in $\operatorname{span}_{\mathbb{R}} \{ \mathbb{J}_1, \mathbb{J}_2, \mathbb{J}_3 \}$.
    \end{enumerate}
In both cases we have
    \begin{align}\label{eq: product properties of the pair-wise anticommutative triple}
         \mathbb{J}_i \mathbb{J}_{i+1} = \mathbb{J}_{i+2} \mod{3} \ , \  i = 1,2,3 \ .
    \end{align}
\end{proposition}

\begin{proof}
Let $\mathbb{A}$ be an $\Rbb$-linear combination of the almost generalized structures $\mathbb{J}_i$
    \begin{align}\label{def: lin. comb. of gen. structures}
       \mathbb{A} = a_1 \mathbb{J}_1 + a_2 \mathbb{J}_2 + a_3 \mathbb{J}_3 \ . 
    \end{align}
From the mutual anticommutativity of the triple, we get
    \begin{align*}
        \mathbb{A}^2 & = \left( {a_1}^2 + {a_2}^2 \epsilon_2  + {a_3}^2 \epsilon_3 \right) (-\mathbb{1}) \ , \\
        \mathbb{A}^* \eta & = ({a_1}^2 + {a_2}^2\epsilon_2 + {a_3}^2 \epsilon_3 ) \eta \ . 
    \end{align*}    
Suppose $\epsilon_1 = \epsilon_2 = \epsilon_3 = 1$. Then for any $(a_1,a_2,a_3) \in  \Rbb^3$ such that $ {a_1}^2 + {a_2}^2 + {a_3}^2 = 1 $ we have $\mathbb{A}^2 = -\mathbb{1} $ and $\mathbb{A}^* \eta = \eta$. Thus $\mathbb{A}$ is a GaC structure. Now suppose $\epsilon_1 = - \epsilon_2 = - \epsilon_3 = 1$. Then 
    \begin{align}
        \mathbb{A}^2 & = \left(-{a_1}^2 + {a_2}^2 + {a_3}^2 \right) \mathbb{1} \ , \\ 
        \mathbb{A}^*\eta & = \left({a_1}^2 - {a_2}^2 - {a_3}^2 \right) \eta \ .
    \end{align}
Thus for $-{a_1}^2 + {a_2}^2 + {a_3}^2 = 1 $ we have $\mathbb{A}^2 = \mathbb{1}$ and $ \mathbb{A}^*\eta = - \eta$, so $\mathbb{A}$ is a GaPC structure. Similarly, for $ -{a_1}^2 + {a_2}^2 + {a_3}^2 = -1 $ we have $\mathbb{A}^2 = -\mathbb{1}$ and $ \mathbb{A}^*\eta = \eta$, hence $\mathbb{A}$ is a GaC structure. The remaining statements of the proposition were discussed before the proposition.
\end{proof}

\begin{definition}[\cite{cortes2021generalized, HONG2015187-from-hypercomplex-to-symplectic}]
\textit{A generalized almost hyper-complex structure} is a~triple $\{ \mathbb{A}, \mathbb{B}, \mathbb{C} \}$ of GaC structures such that $\mathbb{A} \mathbb{B}\mathbb{C} = \mathbb{-1}$. \textit{A generalized almost hyper-para-complex structure} is a triple $\{ \mathbb{A}, \mathbb{B}, \mathbb{C} \}$, consisting of a GaC structure and a pair of GaPC structures such that $\mathbb{A} \mathbb{B} \mathbb{C} = \mathbb{1}$. The structures are integrable, if $\{ \mathbb{A}, \mathbb{B}, \mathbb{C} \}$ are integrable. In this case we call the triple \textit{generalized hyper-complex} and \textit{generalized hyper-para-complex} structure, respectively.
\end{definition}

\begin{remark}
In \cite{HONG2015187-from-hypercomplex-to-symplectic}, the notion of a generalized hyper-complex structure is named \textit{hypercomplex structure on a Courant algebroid} (in our case, the Courant algebroid is the generalized tangent bundle).
\end{remark}

\begin{proposition}
The triplet of pair-wise anticommuting generalized structures given by \eqref{triple of gen. structures 1} satisfies
    \begin{align}
         \mathbb{J}_1 \mathbb{J}_2 \mathbb{J}_3 = \epsilon_3 (\mathbb{-1}) \ .
    \end{align}
    \begin{enumerate}
        \item If $\epsilon_3 = 1$, then $\{ \mathbb{J}_1, \mathbb{J}_2, \mathbb{J}_3 \}$ is a generalized hyper-complex structure on $\mathcal{B}$.  
        \item If $\epsilon_3 = -1$, then $\{ \mathbb{J}_1, \mathbb{J}_2, \mathbb{J}_3 \}$ is a generalized hyper-para-complex structure on $\mathcal{B}$.
    \end{enumerate}
\end{proposition}

\begin{proof}
From proposition \eqref{prop: the 2-sphere of almost GC structures} we know, independently of the value of $\epsilon_3$, that $\mathbb{J}_1 \mathbb{J}_2 = \mathbb{J}_3$ and $\epsilon_1 = 1$, which means that all the three structures are integrable. Due to \eqref{eq: (J_i)^2}, we have $\mathbb{J}_1 \mathbb{J}_2 \mathbb{J}_3 = \epsilon_3 (-\mathbb{1})$. If $\epsilon_3 = 1$, then $\epsilon_2 = 1$, and thus $\{\mathbb{J}_1, \mathbb{J}_2, \mathbb{J}_3 \}$ is a triple of GC structures (cf. table \eqref{tab: type of generalized structure for det hess = 1}). Moreover, $\mathbb{J}_1\mathbb{J}_2\mathbb{J}_3 = \mathbb{-1}$, so the triple is a generalized hyper-complex structure. On the other hand, if $\epsilon_3 = -1$, then $\epsilon_2 = -1$, and thus $\{ \mathbb{J}_1, \mathbb{J}_2, \mathbb{J}_3 \}$ contains a GC structure and a~pair of GPC structures (cf. table \eqref{tab: type of generalized structure for det hess = 1}). Moreover, $\mathbb{J}_1\mathbb{J}_2\mathbb{J}_3 = \mathbb{1}$, so the triple is a~generalized hyper-para-complex structure. 
\end{proof}



\subsection{Generalized metric compatible structures}

We say that a generalized almost product structure $\mathbb{G} \in \operatorname{End} \left( T \mathcal{M} \right)$ is \textit{a generalized (pseudo-)metric}, if the eigenbundles of $\mathbb{G}$ are isomorphic to $T\mathcal{M}$, that is, if $\mathbb{G}$ is non-degenerate. 

\begin{example}\label{example: generalized metric}
The main example of an indefinite generalized metric can be constructed from an indefinite metric \eqref{eq: L-R def} as $\mathbb{G} = \begin{pmatrix} 0 & \pi_g^\# \\ g_\# & 0 \end{pmatrix}$. 
\end{example}

\begin{definition}[\cite{hu2019commuting}]
\textit{A generalized Kähler structure} is a commuting pair of a generalized metric and a GaC structure. \textit{A generalized chiral structure} is a~commuting pair of a generalized metric and a GaP structure. 
\end{definition}

\begin{proposition}\label{prop: generalized Kahler and chiral structure from det Hess = 1}
The M-A equation $\det \operatorname{Hess} (f) = \frac{1}{2}\Delta P$ defines a generalized Kähler structure for $\Delta P = 2$, and a generalized chiral structure for $\Delta P = -2$.
\end{proposition}    
    
\begin{proof}
Let $\mathbb{G}$ and $\mathbb{J}$ be given by the following block matrices
    \begin{align*}
        \mathbb{G} & = 
            \begin{pmatrix}
            0 & \underline{g} \\
            \underline{g} & 0
            \end{pmatrix} \ , & 
        \mathbb{J} & = 
            \begin{pmatrix}
            \underline{\rho}  & 0\\
            0 &  \sgn \left( \Delta P \right)\underline{\rho} 
            \end{pmatrix} \ , 
    \end{align*}
where $\underline{g} = \begin{pmatrix} 0 & \mathbb{1} \\ \mathbb{1} & 0 \end{pmatrix}$ corresponds to \eqref{eq: Lychagin-Rubtsov metric for det hess = 1}, with $|\Delta P| = 2$, and $\rho$ is given by \eqref{eq: Lychagin-Rubtsov metric for det hess = 1}
    \begin{align*}
        \underline{\rho} & =  
            \begin{pmatrix} 
            0 & A \\ 
            \sgn \left( \Delta P \right) A & 0 
            \end{pmatrix} \ , 
    \end{align*}
where $A = \begin{pmatrix}
            0 & -1 \\
            1 & 0
            \end{pmatrix}$ 
Notice that $\mathbb{G}$ and $\mathbb{J}$ correspond to structures from \eqref{diagonal and antidiagonal structures} with $\epsilon = 1$, and where $\rho$ is also allowed to be an almost product structure. Then
    \begin{align}
        [\mathbb{G}, \mathbb{P}] = 0 \quad \iff \quad  
            \begin{cases}
            \ [ \underline{g}, \underline{\rho} ] = 0 & \text{ for } \Delta P = 2 \\
            \{ \underline{g}, \underline{\rho} \} = 0 & \text{ for } \Delta P = -2 \\
            \end{cases}
    \end{align}
A direct computation shows that $[ \underline{g}, \underline{\rho} ] = 0$ is true for $|\Delta P| = 2$ (i.e. independently of the sign of Pfaffian). Thus $\mathbb{G}$ commutes with $\mathbb{J}$ for both values of $\Delta P$. Moreover, $\underline{\rho}$ is a complex structure for $\Delta P = 2$ and a product structure for $\Delta P = -2$, which further implies that $\mathbb{J}$ is a GC structure for $\Delta P = 2$ and a GP structure for $\Delta P = -2$. Hence for $\Delta P = 2$, the pair $(\mathbb{G}, \mathbb{J})$ is a generalized Kähler structure, while for $\Delta P = -2$, the pair is a generalized chiral structure. 
\end{proof}

\subsection{Algebras generated by generalized structures}

There are multiple algebraic structures we can associate to a triple of generalized structures \eqref{triple of gen. structures 1}. The main point of the following paragraph is to discuss the possibility of studying M-A equations and M-A structures on the level of algebraic structures determined by their generalized geometries. The hypothesis which we want to investigate in our future work is whether the algebras defined below are invariant with respect to certain transformations. For example, one might study invariance with respect to symplectomorphisms, which provide notion of equivalence between M-A equations (see proposition \ref{Lychagin-Rubtsov theorem}). 

Let us denote by $\mathcal{A}(\mathbb{J}_i, \epsilon_i)_{i = 1}^3$ the free, associative, unital $\mathbb{R}$-algebra, given by the real vector space generated by the generalized almost structures $\{\mathbb{J}_i\}_{i = 1}^3$ with fixed values $(\epsilon_i)_{i= 1}^3$, and equipped with the product of matrices (or composition of operators, if we work coordinate-free)
    \begin{align}
        \mathcal{A}(\mathbb{J}_i, \epsilon_i)_{i = 1}^3 : = \left( \mathbb{R}[\mathbb{J}_i]_{i = 1}^3, ( \epsilon_i)_{i = 1}^3, \cdot \right) \ . 
    \end{align}
We will omit the $\cdot$ symbol when writing product of elements from $ \mathcal{A}(\mathbb{J}_i, \epsilon_i)_{i = 1}^3 $. We further consider two other bilinear operations on this algebra, the commutator and the anticommutator.

\hfill 

\noindent \textbf{Jordan algebra}. By a Jordan algebra we mean a commutative (possibly non-associative) algebra which satisfies $(AB)(AA) = A(B(AA))$ \cite{Ciaglia2021WhatLA}. From every associative $k$-algebra\footnote{Where $k$ is a field of characteristic not equal to $2$} $(\mathcal{A}, \cdot )$ we can construct a Jordan $k$-algebra by substituting the algebra product with half of the anticommutator
    \begin{align}
        A \bullet B : = \frac{1}{2} \{ A, B\} \ .
    \end{align}
The resulting Jordan algebra is called a \textit{special Jordan algebra}. Let us denote the special Jordan algebra constructed this way from $\mathcal{A}(\mathbb{J}_i, \epsilon_i)_{i = 1}^3$ by $\operatorname{Jor}(\mathbb{J}_i, \epsilon_i)_{i = 1}^3$, i.e. 
    \begin{align}
        \operatorname{Jor}(\mathbb{J}_i, \epsilon_i)_{i = 1}^3 : = \left( \mathbb{R}[\mathbb{J}_i]_{i = 1}^3, ( \epsilon_i)_{i = 1}^3, \bullet \right) \ . 
    \end{align}

\hfill

\noindent \textbf{Lie algebra}. We will denote the Lie algebra generated by the generalized structures $\{\mathbb{J}_i\}_{i = 1}^3$ with fixed epsilons $(\epsilon_i)_{i = 1}^3$ by $\operatorname{Lie}(\mathbb{J}_i, \epsilon_i)_{i = 1}^3$, i.e. 
    \begin{align}
        \operatorname{Lie}(\mathbb{J}_i, \epsilon_i)_{i = 1}^3 : = \left( \mathbb{R}[\mathbb{J}_i]_{i = 1}^3, ( \epsilon_i)_{i = 1}^3, [-,-] \right) \ . 
    \end{align}

\hfill 

\noindent \textbf{Poisson algebra}. We have two Poisson algebra structures at hand. One of them is coming from a Lie algebra, the other one from a Jordan algebra. 

\begin{enumerate}
    \item Since the underlying algebra $ \mathcal{A}(\mathbb{J}_i, \epsilon_i)_{i = 1}^3$ is associative and the Lie bracket is given by the commutator, we automatically have a Poisson algebra at hand 
        \begin{align}
            \operatorname{Poiss}(\mathbb{J}_i, \epsilon_i)_{i = 1}^3 : = \left(  \mathcal{A}(\mathbb{J}_i, \epsilon_i)_{i = 1}^3 ,  [-,-] \right) \ . 
        \end{align}
    Indeed, for every $\mathbb{A}, \mathbb{B}, \mathbb{C} \in  \mathcal{A}(\mathbb{J}_i, \epsilon_i)_{i = 1}^3 $ we have
    \begin{align}
        [\mathbb{A}, \mathbb{B} \mathbb{C} ] = [\mathbb{A}, \mathbb{B}] \mathbb{C} + \mathbb{B} [\mathbb{A}, \mathbb{C}] \ . 
    \end{align}
    \item The second Poisson algebra is given by considering the (non-associative) Jordan algebra $\operatorname{Jor}(\mathbb{J}_i, \epsilon_i)_{i = 1}^3$, equipped with the commutator
        \begin{align}
            \operatorname{Poiss}(\mathbb{J}_i, \epsilon_i, \bullet)_{i = 1}^3 : = \left(\operatorname{Jor}(\mathbb{J}_i, \epsilon_i)_{i = 1}^3, [-,-] \right) \ . 
        \end{align}
    Indeed, for every $\mathbb{A}, \mathbb{B}, \mathbb{C} \in \mathbb{R}[\mathbb{J}_i]_{i = 1}^3$ we have
    \begin{align}
        [\mathbb{A}, \mathbb{B} \bullet \mathbb{C} ] = [\mathbb{A}, \mathbb{B}] \bullet \mathbb{C} + \mathbb{B} \bullet [\mathbb{A}, \mathbb{C}] \ . 
    \end{align}    
\end{enumerate}

\begin{remark}
 In the second case of Poisson algebra comming from a Jordan algebra, one might wonder whether we have a Lie-Jordan algebra at hand, i.e. an algebraic structure which is both, a Lie algebra, a Jordan algebra, and the two bilinear operations are compatible in the following sense
        \begin{align}
            [\mathbb{A}, \mathbb{B} \bullet \mathbb{C} ] & = [\mathbb{A}, \mathbb{B}] \bullet \mathbb{C} + \mathbb{B} \bullet [\mathbb{A}, \mathbb{C}] \ , \\
            q^2 [[\mathbb{A}, \mathbb{C}] , \mathbb{B}] & = (\mathbb{A} \bullet \mathbb{B}) \bullet \mathbb{C} - \mathbb{A} \bullet( \mathbb{B} \bullet \mathbb{C} )\ ,
        \end{align}
    where $q \in \mathbb{R}$ (or $q \in k$, if we consider a different underlying field). As we have stated above, the first axiom is satisfied. \textit{The second axiom is satisfied if, and only if $q = \frac{\sqrt{-1}}{2}$.}  
\end{remark} 

\begin{remark}
V. Lychagin developed in \cite{Lychagin-Diff_eq_on_2D_mflds} all the necessary definitions of algebraic structures, such as Lie algebras and Jordan algebras, defined for a~system of two dimensional non-linear first order partial differential equations with a~specific type of non-linearity. A subclass of these systems are called \textit{Jacobi systems}. In our future work, we will be interested in extension of our algebraic setup to such systems, based on the results in \cite{Lychagin-Diff_eq_on_2D_mflds}. 
\end{remark}


\section{Conclusion}


With the help of Monge-Ampère theory \cite{Lychagin, Lychagin-Diff_eq_on_2D_mflds, kushner_lychagin_rubtsov_2006}, we have constructed a~family of generalized structures associated to a Monge-Ampère equation \eqref{Monge-Ampère equation for the stream function}, which was derived from the Navier-Stokes system for 2D incompressible fluid flows. Distinguishing between the hyperbolic and elliptic case, we have shown that many of these generalized structures are integrable and give rise to further composite structures such as hyper-complex and hyper-para-complex. Introducing an indefinite generalized metric defined with the help of pseudo-metric \eqref{eq: Lychagin-Rubtsov metric for det hess = 1}, we have also found generalized metric compatible structures, namely generalized Kähler and generalized chiral structures. 

Our future work will be focused on the study of other M-A equations and their relation to generalized geometry, with the aim to go into higher dimensions. Particularly interesting seems to be dimension three, as the situation there is much richer, yet the symplectic classification of differential forms is still possible. In this context, we want to further investigate the correspondences between M-A equations, families of generalized structures, and algebras constructed from these structures. We are especially curious which kind of equivalence notions naturally present in the M-A theory are able to survive when passing to either generalized structures, or various algebras we have linked with them. Yet another direction to investigate is the possiblity of extending this framework to the Jacobi systems and other more general situations of first or second order PDEs with special types of non-linearities, since a framework for doing so has already been developed \cite{Lychagin-Diff_eq_on_2D_mflds}.

\section*{Acknowledgments}

First of all I want to express my wholehearted gratitude to my beloved wife and son for their continuous support and patience. This paper was written as a part of my PhD research under the cotutelle agreement between the Masaryk University, Brno, Czech Republic, and the University of Angers, France. I am grateful for hospitality and support of these institutions during the research period. I~am also grateful for the funding provided by the Czech Ministry of Education, and by the Czech Science Foundation under the project GAČR EXPRO GX19-28628X. I want to express my gratefulness to my supervisors Volodya Rubtsov and Jan Slovák for their supervision over my research project. I especially want to thank Volodya Rubtsov for his valuable remarks and suggestions of improvement of this work. The results of this paper were reported at the Winter School and Workshop Wisla 22, a European Mathematical Society event organized by the Baltic Institute of Mathematics.


\bibliographystyle{plain}

\begin{thebibliography}{10}

\bibitem{VolRoul2015}
B.~Banos, V.~Rubtsov, and I.~Roulstone.
\newblock {Monge--Amp{\`e}re Structures and the Geometry of Incompressible
  Flows}.
\newblock {\em Journal of Physics A: Mathematical and Theoretical}, 49, 10
  2016.

\bibitem{BANOS2007841}
Bertrand Banos.
\newblock {Monge–Ampère equations and generalized complex geometry— The
  two-dimensional case}.
\newblock {\em Journal of Geometry and Physics}, 57(3):841--853, 2007.

\bibitem{BANOS20112187}
Bertrand Banos.
\newblock Complex solutions of monge–ampère equations.
\newblock {\em Journal of Geometry and Physics}, 61(11):2187--2198, 2011.

\bibitem{Ciaglia2021WhatLA}
Florio~M. Ciaglia, Jurgen Jost, and Lorenz~J. Schwachhofer.
\newblock What lie algebras can tell us about jordan algebras.
\newblock {\em arXiv.org}, 2021.

\bibitem{cortes2021generalized}
Vicente Cortés and Liana David.
\newblock Generalized connections, spinors, and integrability of generalized
  structures on courant algebroids, 2021.

\bibitem{Courant-DiracBracket}
Theodore~James Courant.
\newblock Dirac manifolds.
\newblock {\em Transactions of the American Mathematical Society},
  319(2):631--661, 1990.

\bibitem{Crainic2004GeneralizedCS}
Marius Crainic.
\newblock Generalized complex structures and lie brackets.
\newblock {\em Bulletin of the Brazilian Mathematical Society, New Series},
  42:559--578, 2004.

\bibitem{phdthesis}
Sylvain Delahaies.
\newblock {\em Complex and contact geometry in geophysical fluid dynamics}.
\newblock PhD thesis, University of Surrey (United Kingdom), 01 2008.

\bibitem{dritschel_viudez_2003}
DAVID~G. DRITSCHEL and ÁLVARO VIÚDEZ.
\newblock A balanced approach to modelling rotating stably stratified
  geophysical flows.
\newblock {\em Journal of Fluid Mechanics}, 488:123–150, 2003.

\bibitem{Gualtieri2007BranesOP}
Marco Gualtieri.
\newblock Branes on poisson varieties.
\newblock {\em The many facets of geometry}, page 368–394, 2010.

\bibitem{Gualtierri2011}
Marco Gualtieri.
\newblock Generalized complex geometry.
\newblock {\em Annals of Mathematics}, 174(1):75--123, 2011.

\bibitem{Hitchin-gen-calabi-yau}
Nigel Hitchin.
\newblock {Generalized Calabi–Yau Manifolds}.
\newblock {\em The Quarterly Journal of Mathematics}, 54(3):281--308, 09 2003.

\bibitem{HONG2015187-from-hypercomplex-to-symplectic}
Wei Hong and Mathieu Stiénon.
\newblock From hypercomplex to holomorphic symplectic structures.
\newblock {\em Journal of Geometry and Physics}, 96:187--203, 2015.

\bibitem{HronekSuchanek}
Stanislav Hronek and Radek Suchánek.
\newblock {Pseudo-Riemannian and Hessian geometry related to Monge-Amp{\`e}re
  structures}.
\newblock {\em Archivum Mathematicum}, 58(5):329--338, 2022.

\bibitem{hu2019commuting}
Shengda Hu, Ruxandra Moraru, and David Svoboda.
\newblock Commuting pairs, generalized para-k$\backslash$" ahler geometry and
  born geometry.
\newblock {\em arXiv preprint arXiv:1909.04646}, 2019.

\bibitem{Kosmann-Schwarzbach2010}
Yvette Kosmann-Schwarzbach and Vladimir Rubtsov.
\newblock {Compatible Structures on Lie Algebroids and Monge-Amp{\`e}re
  Operators}.
\newblock {\em Acta Applicandae Mathematicae}, 109(1):101--135, Jan 2010.

\bibitem{kushner_lychagin_rubtsov_2006}
Alexei Kushner, Valentin Lychagin, and Vladimir Rubtsov.
\newblock {\em Contact Geometry and Nonlinear Differential Equations}.
\newblock Encyclopedia of Mathematics and its Applications. Cambridge
  University Press, 2006.

\bibitem{Lychagin}
V.~V. {Lychagin}.
\newblock {Contact Geometry and Non-Linear Second-Order Differential
  Equations}.
\newblock {\em Russian Mathematical Surveys}, 34(1):149--180, February 1979.

\bibitem{Lychagin-Diff_eq_on_2D_mflds}
V.~V. Lychagin.
\newblock {Differential equations on two-dimensional manifolds}.
\newblock {\em Izv. Vyssh. Uchebn. Zaved. Mat.}, 5:43--57, 1992.

\bibitem{Lychagin-Rubtsov-1983}
V.~V. Lychagin and V.~N. Rubtsov.
\newblock Local classification of {Monge}-{Amp{\`e}re} differential equations.
\newblock {\em Sov. Math., Dokl.}, 28:328--332, 1983.

\bibitem{Lychagin-Rubtsov-Chekalov}
Valentin~V. Lychagin, V.~N. Rubtsov, and I.~V. Chekalov.
\newblock A classification of {Monge-Amp\`ere} equations.
\newblock {\em Annales scientifiques de l'\'Ecole Normale Sup\'erieure}, Ser.
  4, 26(3):281--308, 1993.

\bibitem{mcintyre2001there}
Michael~E McIntyre and Ian Roulstone.
\newblock {\em Are there higher-accuracy analogues of semigeostrophic theory?},
  volume~2 of {\em Large-scale atmosphere-ocean dynamics}.
\newblock Cambridge Univ. Press, 2002.

\bibitem{coherent_structures_Navier-Stokes}
I.~Roulstone, B.~Banos, J.~D. Gibbon, and V.~N. Roubtsov.
\newblock {A geometric interpretation of coherent structures in Navier-Stokes
  flows}.
\newblock {\em Proc. R. Soc. Lond. Ser. A Math. Phys. Eng. Sci.},
  465(2107):2015–2021, 2009.

\bibitem{Roulstone-Sewell}
I.~Roulstone and M.~J. Sewell.
\newblock The mathematical structure of theories of semigeostrophic type.
\newblock {\em Philosophical Transactions of the Royal Society of London.
  Series A: Mathematical, Physical and Engineering Sciences},
  355(1734):2489--2517, 1997.

\bibitem{roulstone2009kahler}
Ian Roulstone, Bertrand Banos, J.D. Gibbon, and V.N. Roubtsov.
\newblock {Kähler geometry and Burgers' vortices}.
\newblock {\em Proceedings of Ukrainian National Academy Mathematics},
  16(3):303--321, 2009.

\bibitem{RoubRoul2001}
V.~Rubtsov and I.~Roulstone.
\newblock Holomorphic structures in hydrodynamical models of nearly geostrophic
  flow.
\newblock {\em Proc. R. Soc. Lond. A.}, 457:1519–1531, 06 2001.

\bibitem{Rubtsov_1997}
V.~N. Rubtsov and I.~Roulstone.
\newblock {Examples of quaternionic and Kähler structures in Hamiltonian
  models of nearly geostrophic flow}.
\newblock {\em Journal of Physics A: Mathematical and General}, 30(4):L63--L68,
  feb 1997.

\bibitem{Rubtsov2019}
Volodya Rubtsov.
\newblock {\em Geometry of Monge-Ampère structuress}.
\newblock Tutor. Sch. Workshops Math. Sci. Birkhäuser/Springer, Cham,, 2019.

\bibitem{viudez_dritschel_2002}
ÁLVARO VIÚDEZ and DAVID~G. DRITSCHEL.
\newblock An explicit potential-vorticity-conserving approach to modelling
  nonlinear internal gravity waves.
\newblock {\em Journal of Fluid Mechanics}, 458:75–101, 2002.

\end{thebibliography}

\end{document}